\documentclass[english,11pt]{article}
\pdfoutput=1
\usepackage{amsthm, amssymb, fullpage, graphicx, subfigure, charter,hyperref}

\newtheorem{theorem}{Theorem}
\newtheorem{lemma}{Lemma}
\newtheorem{corollary}{Corollary}
\newtheorem{observation}{Observation}
\newcommand{\R}{\mathbb R}
\newcommand{\bigc}{\mathcal C}
\newcommand{\bigt}{\mathcal T}
\newcommand{\bigw}{\mathcal W}
\newcommand{\tr}{{\mathcal T}} 

\renewcommand{\bar}[1]{\textstyle\widetilde{\scriptstyle{#1}}}
\newcommand{\closure}{\mathrm{Cl}}

\title{Decomposition of Multiple Coverings into More Parts}
\author{Greg Aloupis\thanks{Universit\'e Libre de Bruxelles (ULB), CP212, Bld. du Triomphe, 1050 Bruxelles, Belgium. Supported by the \emph{Communaut\'e Fran\c caise de Belgique}. E-mail: {\tt \{greg.aloupis,jcardin,secollet,slanger\}@ulb.ac.be}.} \and Jean Cardinal\footnotemark[1] \and  S\'ebastien Collette\footnotemark[1]~\thanks{Charg\'e de Recherches du FRS-FNRS.} \and Stefan Langerman\footnotemark[1]~\thanks{Chercheur Qualifi\'e du FRS-FNRS.} \and David Orden\thanks{Universidad de Alcal\'a, Spain. E-mail: {\tt \{david.orden,pedro.ramos\}@uah.es}.}  \and Pedro Ramos\footnotemark[4]}
\date{}

\begin{document}
\maketitle
\sloppy

\begin{abstract}
We prove that for every centrally symmetric convex polygon $Q$, there exists a constant $\alpha$ such that any $\alpha k$-fold covering of the plane by
translates of $Q$ can be decomposed into $k$ coverings. This improves on a quadratic upper bound proved by Pach and T\'oth (SoCG'07). The question is motivated by a sensor network problem, in which a region has to be monitored by sensors with limited battery lifetime.
\end{abstract}

\section{Introduction}

A collection of subsets of the plane forms a {\em $f$-fold covering} if any point in the plane is covered by at least $f$ subsets. We consider the following problem (see Figure~\ref{fig:decomp}):\\
{\em Given a convex planar body $Q$, does there exist a function $f(Q,k)$ such that any $f(Q,k)$-fold covering of the plane by translates of $Q$ can be decomposed into $k$ disjoint (1-fold) coverings?}\medskip

This problem, first raised by Pach and T\'oth in 1980 (see~\cite{Pa80} and references therein), is  a classical question in discrete geometry and remains largely open. In fact it is not even known whether there exists a constant $c$ such that any $c$-fold covering  can be decomposed into {\em two} coverings. A survey of the literature can be found in the book of Brass, Moser, and Pach~\cite{RPDG}.\medskip

\begin{figure}[htb]
\begin{center}
\includegraphics[scale=.5,angle=-90]{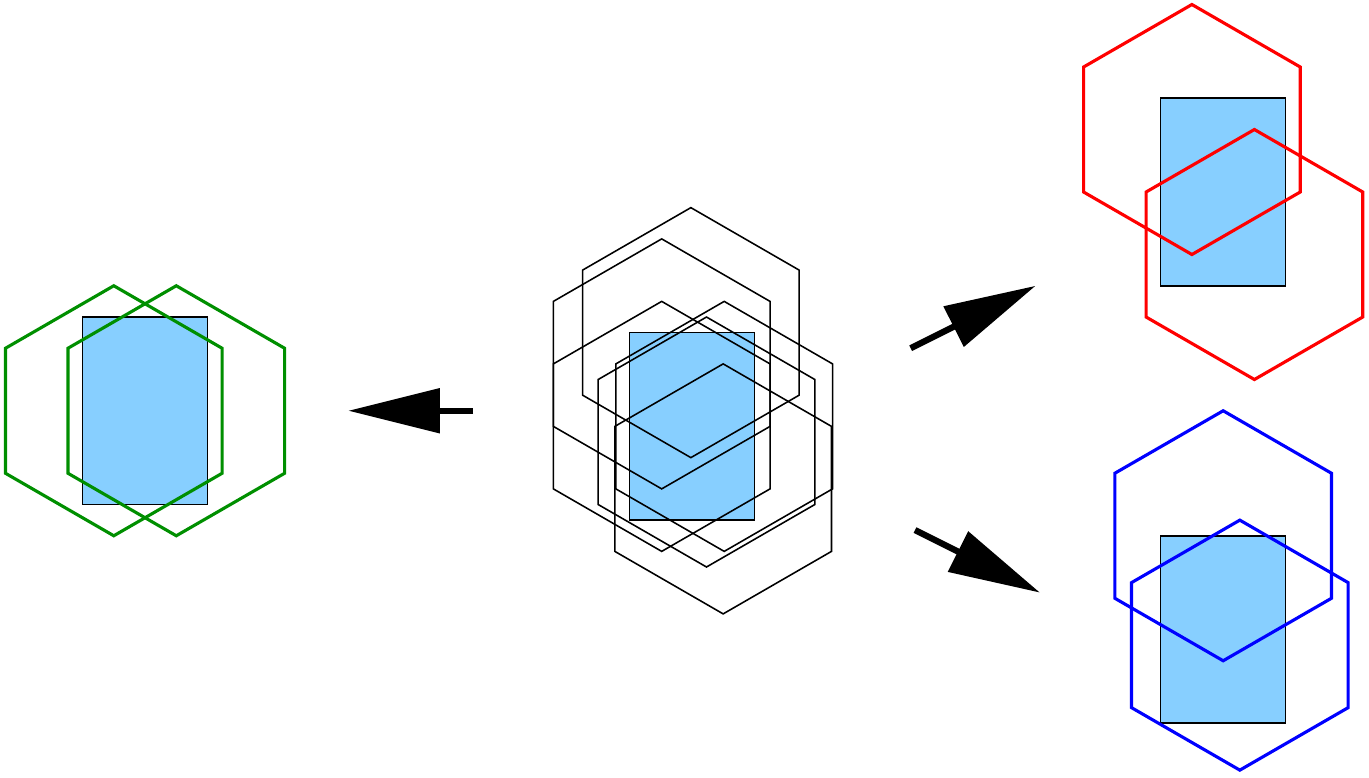}
\end{center}
\caption{\label{fig:decomp}A $3$-fold covering of a rectangle by hexagons that can be decomposed into three coverings.}
\end{figure}

Mani and Pach proved that 33-fold coverings by unit disks can be decomposed into two coverings~\cite{MP86}. Tardos and T\'{o}th recently proved that, for triangles, any 43-fold covering can be decomposed into two coverings~\cite{TT07}. 

For the case of centrally symmetric convex polygons, the problem proved to be challenging. The existence of a function $f(Q,k)$
was conjectured in 1980~\cite{Pa80}, and a few years later, resolved positively~\cite{Pach86} by Pach. Only twenty years later was it shown that $f(Q,k)$ is at most quadratic in $k$.

\begin{theorem}[Pach and T\'{o}th~\cite{PT07}]
\label{we-beat-pach}
Given a centrally symmetric convex polygon $Q$, there exists a constant $\alpha_Q$ such that every $\alpha_Q k^2$-fold covering of the plane by translates of $Q$ can be decomposed into $k$ coverings.
\end{theorem}

\noindent  In addition to the above, a lower bound of  $\lfloor {4k/3}\rfloor{-}1$ was also given.\medskip

The main result  in this paper is an improvement of the bound in Theorem~\ref{we-beat-pach}, from $\alpha_Q k^2$ to $\alpha_Q k$; thus, the upper and lower bounds now asymptotically match.

\paragraph{Related Work.} 
Coverings with other families of convex shapes have also been studied. For instance, indecomposable coverings of the plane by strips and rectangles were given by Pach, Tardos and T\'{o}th~\cite{PTT07}. The problem for arbitrary disks remains open, although a negative result for the dual problem was proved in~\cite{PTT07}: for any $k$, there exists a point set such that for any 2-coloring of this set, an open disk containing $k$ points of the same color can be found. Set-theoretic investigations of infinite-fold coverings can be found in~\cite{settheo}.

Note that covering decompositions can be seen as colorings of geometric hypergraphs. In these hypergraphs, vertices are the convex bodies in the covering, and every point in the plane corresponds to a hyperedge, defined as the set of bodies containing that point.
The assignment of colors to the vertices of this graph, such that every hyperedge contains all $k$ colors, yields a suitable decomposition.
A recent study of such problems and of their dual, including colorings of hypergraphs induced by halfspaces, halfplanes,  disks, and pseudo-disks, is contained in~\cite{ACCLS08}.\footnote{Using the notation of~\cite{ACCLS08}, the main result of this paper is that $p_{\bar{\tr}}(k) = O(k)$.}\medskip

Other definitions of proper colorings of geometric hypergraphs have been studied, such as \emph{conflict-free} colorings~\cite{shakharcf}. 
Here the problem is to find a coloring such that every hyperedge contains at least one vertex with a \emph{unique} color.
Variants of this notion have also been analyzed, e.g., \emph{$k$-fault-tolerant} conflict-free colorings where the conflict-free property must be true even if we were to remove any $k$ vertices in a hyperedge~\cite{faulttolerant}. \emph{$k$-conflict-free} colorings~\cite{jithamilton} require $k$ vertices with unique colors in every hyperedge. 

\paragraph{Applications to sensor networks.}
Consider a planar region monitored by sensors. Each sensor is represented as a point, which is said to monitor every other point contained in a polygonal region around it. Sensors are assumed to have limited lifetime, but can be switched on at any chosen time. Such models of limited-lifetime sensors have been studied in other contexts~\cite{othersensors}. Our results imply that a region can be monitored for $k$ units of time, provided that every point is covered by at least $\alpha k$ sensors. This involves partitioning the set of sensors into $k$ subsets, each covering the region. Sensors in the 
$j$-th subset are switched on at time $j$. 

\paragraph{Problem modification.}
We now slightly modify the statement of the problem.  Let $Q_p$ denote a centrally symmetric polygon $Q$ centered at point $p$.  Notice that $Q_p$ covers a point $p'$ if and only if $Q_{p'}$ contains $p$.
 
The problem involves a set of translates of $Q$ that covers every point of the plane  at least $\alpha_Q k$ times. 
This is geometrically equivalent to a point set $S$ such that any translate of $Q$ in the plane contains at least $\alpha_Q k$ points of $S$.    
Note, of course that $S$ must be infinite, as must be $Q$ in the original problem.

The decomposition of translates into $k$ covers is equivalent to a coloring of $S$ such that every translate of $Q$ in the plane will contain $k$ colors.
We strengthen the problem statement, by relaxing the condition that every translate in the plane contains sufficiently many points.  That is, we say that if a translate contains enough points, it will contain $k$ colors. This allows us to consider finite point sets as well. 
We thus prove the following result.

\begin{theorem}
\label{thm:rephrased}
Given a centrally symmetric convex polygon $Q$, there exists a constant $\alpha_Q$ such that for every planar point set $S$ and every $k\in \mathbb N$, 
$S$ can be $k$-colored so that any translate of $Q$ containing at least $\alpha_Q k$ points will also contain at least one point of each color.
\end{theorem}

For simplicity of exposition, we assume general position: no two
points in $S$ have the same slope as an edge of $Q$. This assumption
can be removed by applying an infinitesimal perturbation to the
points. Also, we assume $S$ is locally finite; 
every compact region contains a finite number of points.

\paragraph{Overview.}
We start by giving a sketch of the complete proof before going into details. The original problem is transformed as follows.

The problem of coloring a (possibly infinite) point set with respect to translates of a polygon (the strengthened statement presented in Theorem~\ref{thm:rephrased}) is  shown to be equivalent to coloring a finite point set  with respect to a finite set of wedges determined by $Q$ (see Section~\ref{sec:prel}). 
In other words,  the problem is now to color a set of points such that every wedge containing a sufficient number of  points $m$ will also contain $k$ colors.  Our goal is to show that $m=O(k)$.
 
We will restrict to color points inside certain \emph{witness} wedges, which have the property that any wedge containing at least $m$ points will contain a witness. Witnesses will contain at least $r$ points. This is why we will define the \emph{level curve} which bounds the union of such minimal wedges for a fixed pair of bounding directions  (see Section~\ref{sec:prel}).

If the level curves did not intersect, coloring the points would be straightforward. It is the intersections of these curves that make the problem non-trivial, and forbids us to restrict to witnesses on level curves only.  Since $k$ is small with respect to the point set, intuitively one can imagine that level curves tend not to venture too ``deep'' into a point set.  In other words, a typical wedge will not reach far into the set before collecting $r$ points.
In Section~\ref{sec:truncation} we define a polygonal region that is deep enough so that the complexity of level curve intersections  within the region is manageable.  Our construction of this region will be such that we will be able to restrict to considering witness wedges within.

To enable us to reduce our problem to circular arc coloring, in Section~\ref{sec:wedges} we define a parametrization which maps the set of witness wedges to the boundary of a circle.  This is directly tied to a mapping of points in $S$ to circular arcs, i.e., intervals on the boundary of the circle (Section~\ref{sec:inter}). Our mapping is such that a position $x$ on the circle will belong to an interval corresponding to point $p\in S$ if and only if the witness wedge represented by $x$ contains $p$.  As every witness wedge contains at least $r$ points, every position on the circle belongs to at least $r$ intervals. The key property of the parametrization is that every point in $S$ is mapped to at most two intervals. 

Thus, the problem is reduced to $k$-coloring arcs on an $\Theta(k)$-covered circle (with certain geometric constraints for the arcs), so that every position on the circle is covered by at least one interval of each color. In Section~\ref{sec:coloring} we give an algorithm for  this circular arc coloring problem. 

 Note that the reductions and the transformations of the problem are constructive; thus our algorithm to color circular arcs yields a simple polynomial algorithm for the original problem. 

\section{Reduction to Wedges}\label{sec:prel}

Let $Q$ be a closed, convex,  centrally symmetric
$2n$-gon, with vertices $q_0,q_1,\ldots , q_{2n{-}1}$ in counterclockwise
order. Throughout the paper, indices are taken modulo $2n$. 
The set of indices between $i$ and $j$ in counterclockwise order is denoted by $[i,j]$.\medskip

We first  reduce the problem to coloring a finite set of points with respect to {\em wedges} instead of coloring a possibly infinite set with respect to polygons. 
This idea is also used in~\cite{PT07,TT07}.

We consider a tiling of the plane, with squares of side
$\delta$, where $\delta$ is half of the smallest distance
between non-consecutive edges of $Q$.
Let $Q'$ be a translate of $Q$.
By construction, any intersection of $Q'$ with a  square is a
wedge
with boundary directions parallel to two consecutive edges of $Q$
(see Figure~\ref{fig:reduc}).
A wedge bounded by rays parallel to $q_iq_{i{-}1}$ and
$q_iq_{i{+}1}$ will be called  {\em type} $i$, or alternatively an {\em $i$-wedge}. 
The closed $i$-wedge with apex $x$ is denoted by $W_i(x)$. 

\begin{figure}[htb]
\begin{center}
\includegraphics[scale=.7]{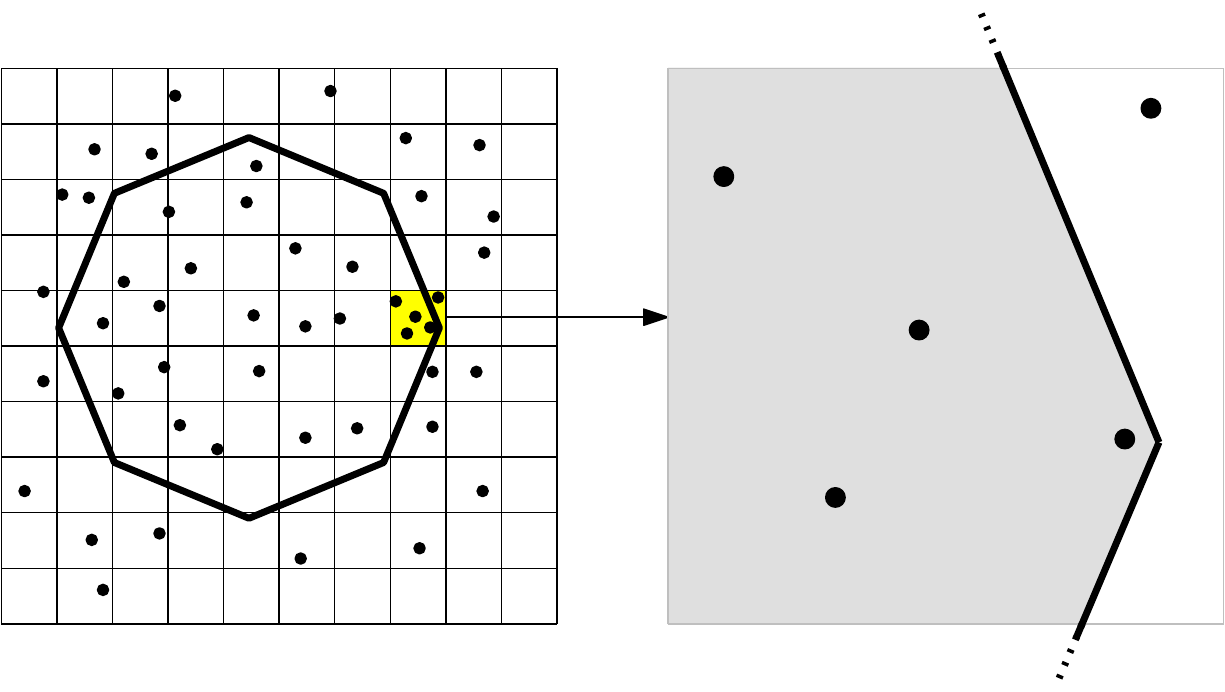}
\end{center}
\caption{\label{fig:reduc} Reduction of the problem with centrally symmetric polygons to wedges in a square.}
\end{figure}

The number of squares that $Q'$ intersects is bounded by
a constant $c_Q$ that only depends on $Q$. Therefore if $Q'$ contains at least $\alpha_Q k$ points,  by the
pigeonhole principle $Q'$ contains at least $\alpha_Q k/c_Q$ points within one square.

We will restrict to considering a single square and
the $2n$ wedges defined by $Q$. Hence the problem reduces to coloring (independently) the finite bounded point set $S$ in each square, i.e., we will seek a $k$-coloring of each square such that any $i$-wedge
containing at least $\alpha_Q k /c_Q$ points will contain all $k$ colors.\medskip

We now define the notion of \emph{level curves} for wedges. This notion extends the 
definition of {\em boundary points} in~\cite{PT07} and~\cite{Pach86}, which are the points found on the first level.
We associate a curve with each $i$-wedge. 
Let $\bigw_i^{r}$ be the set of apices of all $i$-wedges containing $r$ points. Formally,
$$
\bigw_i^{r} := \closure\left(\{x\in \R^2: |W_i(x)\cap S| = r\}\right) ,
$$
\noindent where $\closure(\cdot)$ is the closure operator. We define $\bigc_i(r)$ as the boundary of $\bigw_i^{\geq r}:=\bigcup_{j\ge r}\bigw_i^{j}$.   Accordingly, the closed region that includes the complement of  $\bigw_i^{\geq r}$ will be denoted $\bigw_i^{< r}$ (i.e. the intersection of the two regions is $\bigc_i(r)$).\medskip

Note that $\bigc_i (r)$ is a monotone staircase polygonal path, with edge directions
parallel to those of its corresponding $i$-wedge. Since $S$ is in general
position,  for any $x\in \bigc_i(r)$ that is not a vertex of $\bigc_i(r)$,
$W_i(x)$ contains exactly $r$ points. More precisely, we have the following.
\begin{observation}
\label{obs:rorrp1}
For all $x\in \bigc_i(r)$, $W_i(x)$ contains either $r$ or $r{+}1$ points of $S$.
\end{observation}
\noindent The curves $\bigc_i(3)$ for a square are illustrated in Figure~\ref{fig:curves}. A key property of $\bigc_i$ is the following.
\begin{observation}
Any $i$-wedge containing at least $r$ points of $S$ contains an $i$-wedge whose apex belongs to $\bigc_i (r)$.
\end{observation}
We conclude that it is sufficient to color points in the union of 
all $\bigw_i^{< r}$ (in other words, in the union of regions to the ``left'' of each $\bigc_i(r)$).   Handling the complexity of the intersections of these curves is the next problem that we deal with.

\begin{figure}[htb]
\begin{center}
\includegraphics[angle=-90, scale=.5]{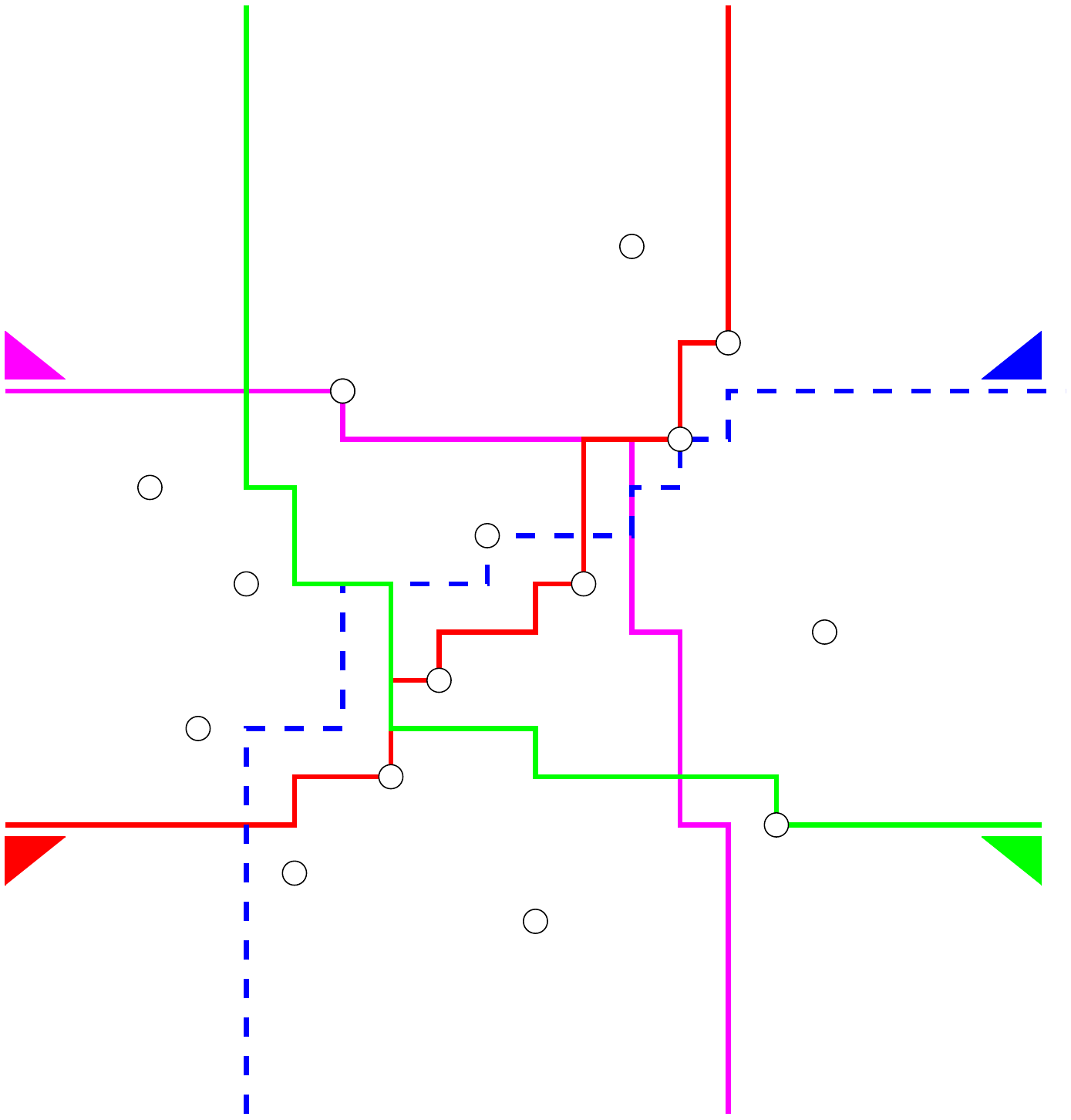}
\end{center}
\caption{\label{fig:curves} The curves of $\bigc_i (3)$, when $Q$ is an axis-parallel square.}
\end{figure}

\section{Restriction to High-Depth Region}
\label{sec:truncation}
We will show that in order to determine the witness wedges that we must  color, it is not necessary to consider complete level curves. At the expense of a constant factor to $f(Q,k)$, we restrict to the portion of the level curves inside a polygon $\bigt$. Inside this polygon, only few intersections between level curves can occur, which simplifies the coloring task.\medskip 

Let $\ell_i$ be the oriented line with
direction $q_iq_{i{+}1}$ going through a point of~$S$
and such that the closed halfplane to its left
contains exactly $2r{+}3$ points.
Let $L_i$ be the closed halfplane to the right of $\ell_i$.
Denote by $\bigt$  the intersection of the $2n$ halfplanes defined by $Q$:
$$
\bigt := \bigcap_{i=0}^{2n{-}1} L_i.
$$
We assume $\bigt\not=\emptyset$: this will be shown true later for the
values of $r$ that we will use (by the well-known center point theorem, it is true as long as $2r{+}3\le |S|/3$). Note that not all lines $\ell_i$
appear on the boundary of $\bigt$ (see Figure~\ref{fig:vis}).

\begin{lemma}
\label{lem:vis}
For all $i\in[0,2n{-}1]$ there is a vertex $v_i$ of $\bigt$ such that
$v_i\in W_i(x)$ for all $x\in\bigt$.
\end{lemma}
\begin{proof}
Let $\hat\ell_i$ be the oriented line parallel to $\ell_i$ that is tangent to~$\bigt$ and such that $\bigt$ is contained in the closed halfplane to the {\em right} of $\hat\ell_i$.
Then for 
$$
v_i := \hat\ell_i\cap\hat\ell_{i{-}1},
$$ 
the wedge~$W_{i{+}n}(v_i)$ contains~$\bigt$. Therefore, $v_i\in W_i(x)$ for all $x\in\bigt$.
Note  that a vertex of $\bigt$ may have multiple labels $v_i$ (see Figure~\ref{fig:vis}).
\end{proof}

\begin{lemma}\label{lem:twowedges}
Let $x$ be a point contained in two wedges $W_i(y)$ and
$W_j(z)$ that  contain at most $r$ and $r'$ points of $S$ respectively, with $0<(j{-}i)<n$. 
Then for all $i'\in [i,j{-}1]$, the oriented line
with direction $q_{i'}q_{i'{+}1}$ through $x$ has at most $r{+}r'$ points
of $S$ strictly to its left.
\end{lemma}
\begin{proof}
It suffices to observe that the halfplane to the left of the line is contained in the union of  wedges $W_i(x)$ and $W_j(x)$.
(See Figure~\ref{fig:twowedges}.)
\end{proof}

\begin{figure}[htb]
\begin{center}
\subfigure[\label{fig:vis}Definition of the points $v_i$.]{\includegraphics[scale=1]{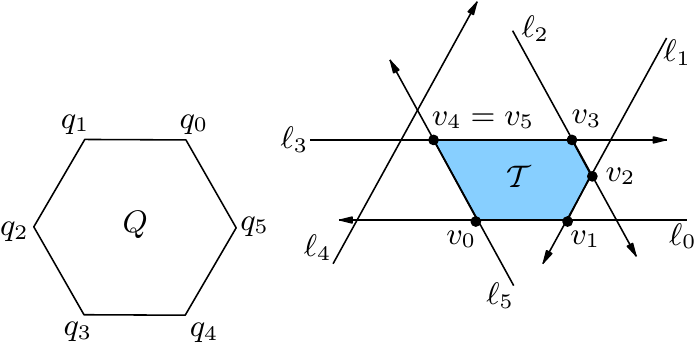}}\hskip 1cm\subfigure[\label{fig:twowedges}Two wedges containing the same point $x$.]{\includegraphics[scale=.5]{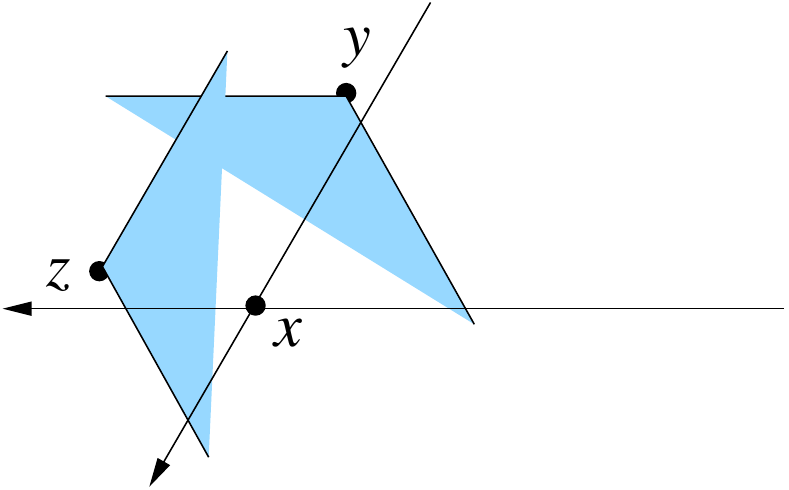}}
\end{center}
\caption{Construction of $\bigt$, and illustration of Lemma~\ref{lem:twowedges}.}
\end{figure}

We now show that if two level curves have an intersection in $\bigt$, then they must have antipodal indices, that is, $i$ and $i{+}n$. 
We actually prove the stronger statement that the regions $\bigw_i^{< r}$ do not have any intersection
in $\bigt$, unless they have antipodal indices.
Note that by Lemma~\ref{lem:vis}, 
$v_{i}\in \bigw_{i}^{<r}$.

\begin{lemma}
\label{lem:antipodal}
If $j\not= i$ and $j\not= i{+}n$, then $\bigw_i^{< r} \cap \bigw_j^{< r}\cap\bigt = \emptyset$.
\end{lemma}
\begin{proof}
Assume by symmetry that $0<(j{-}i)<n$ and suppose the two regions intersect
at point $x\in\bigt$. Consider the two wedges $W_i(x)$ and $W_j(x)$. Since $x$ is contained 
in $\bigw_i^{< r} \cap \bigw_j^{< r}$, they both contain at most $r{+}1$ points. 
By Lemma~\ref{lem:twowedges}, for all $i'\in[i,j{-}1]$, the oriented line
with direction $q_{i'}q_{i'{+}1}$ through $x$ has at most $2r{+}2$ points of $S$ strictly to its left.
This contradicts the fact that $x\in\bigt$.
\end{proof}

We proceed to show that in fact only one pair of level curves can intersect inside $\bigt$ (a related statement was proved by Pach~\cite{Pach86}). 
This is illustrated in Figure~\ref{fig:cprimes}.\bigskip

\begin{figure}[htb]
\begin{center}
\subfigure[A case where the regions $\bigw_i^{< r} \cap \bigt$ (in dark) do not intersect.]{\includegraphics[angle=-90,scale=.4]{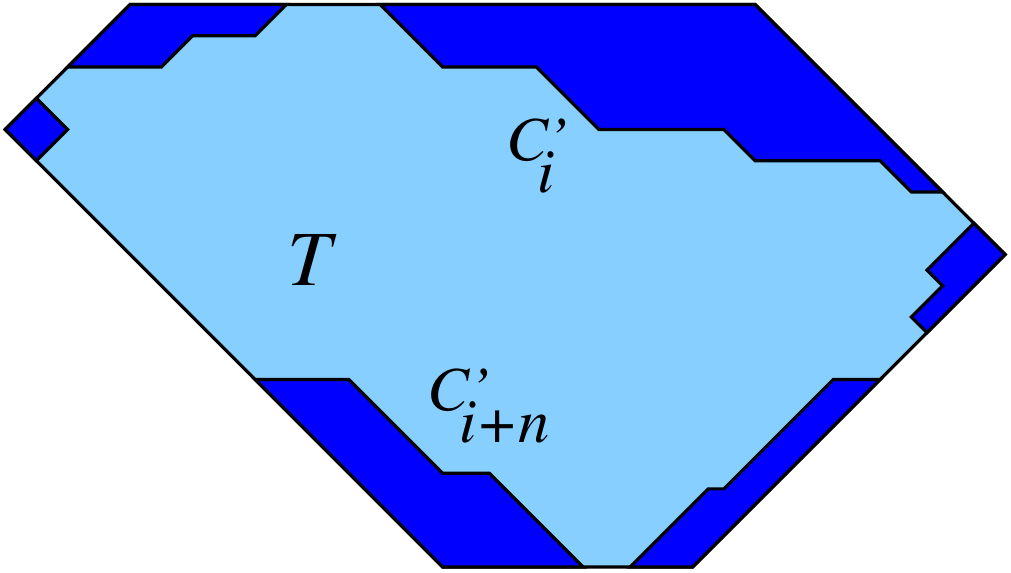}}
\hspace{2cm}
\subfigure[Regions $\bigw_i^{< r}$ and $\bigw_{i{+}n}^{< r}$ may intersect in $\bigt$.]{\includegraphics[angle=-90,scale=.4]{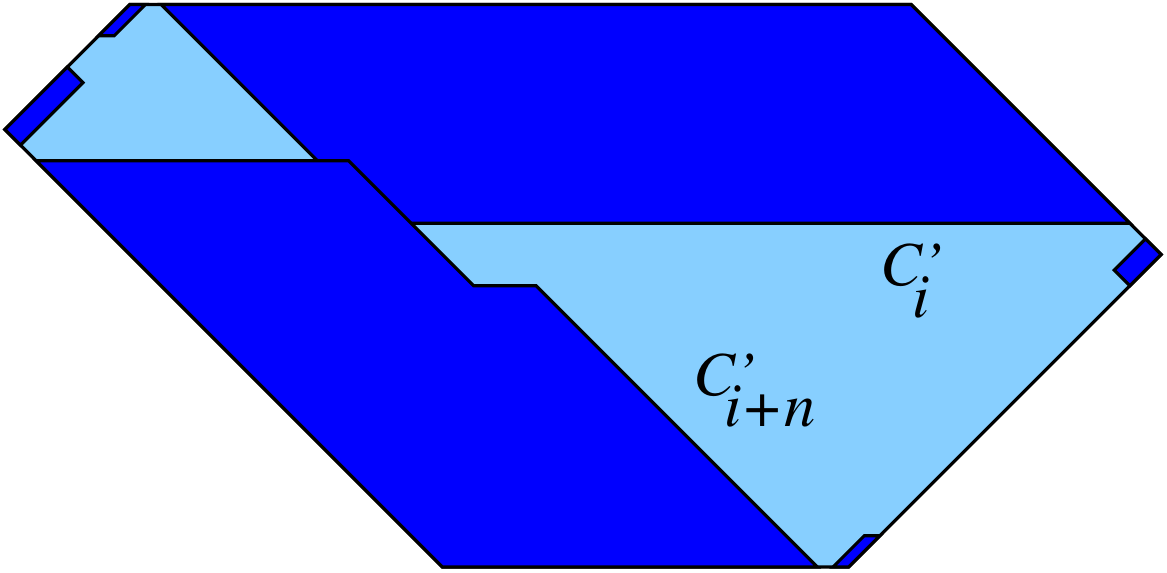}}
\end{center}
\caption{\label{fig:cprimes}Illustration of Lemmas~\ref{lem:antipodal} and~\ref{lem:oneinter}.}
\end{figure}

\begin{lemma}
\label{lem:oneinter}
At most one pair of regions $\{ \bigw_i^{< r} , \bigw_{i{+}n}^{< r}\}$ intersect in $\bigt$.
\end{lemma}
\begin{proof}
By contradiction, suppose that $y\in\bigw_i^{<r}\cap\bigw_{i{+}n}^{<r}\cap\bigt$ and 
$z\in\bigw_j^{<r}\cap\bigw_{j{+}n}^{<r}\cap\bigt$, with $j\not= i{+}n$.

First, let us suppose that $z\in W_i(y)\cup W_{i{+}n}(y)$, and focus on the case $z\in W_i(y)$. Trivially, $z\in W_j(z)$. Thus Lemma~\ref{lem:twowedges} implies that for all $i'\in[i,j{-}1]$, the oriented line with direction $q_{i'}q_{i'{+}1}$ through~$z$ has at most~$2r{+}2$ points of~$S$ strictly to its left, contradicting $z\in\bigt$. The case $z\in W_{i{+}n}(y)$ works analogously.

On the other hand, if $z\not\in W_i(y)\cup W_{i{+}n}(y)$, then we claim that $y\in W_j(z)\cup W_{j{+}n}(z)$ and a similar argument leads to a contradiction. 
In order to prove the claim, consider the rays from $y$ parallel to $q_sq_{s{+}1}$ for $s\in [0,2n-1]$. Then, $z\not\in W_i(y)\cup W_{i{+}n}(y)$ implies that $y$ lies (counterclockwise) between either the pair of rays parallel to $q_{i{-}1}q_{i}$ and to $q_{i}q_{i{+}1}$, or the pair parallel to $q_{i{+}n{-}1}q_{i{+}n}$ and to $q_{i{+}n}q_{i{+}n{+}1}$. Given $j\in[i{+}1,i{+}n{-}1]$, we have $y\in W_{j{+}n}(z)$ in the first case, and $y\in W_j(z)$ in the second case (see Figure~\ref{fig:pflemoneinter}).
\end{proof}

\begin{figure}[htb]
\begin{center}
\includegraphics[scale=.4, angle=-90]{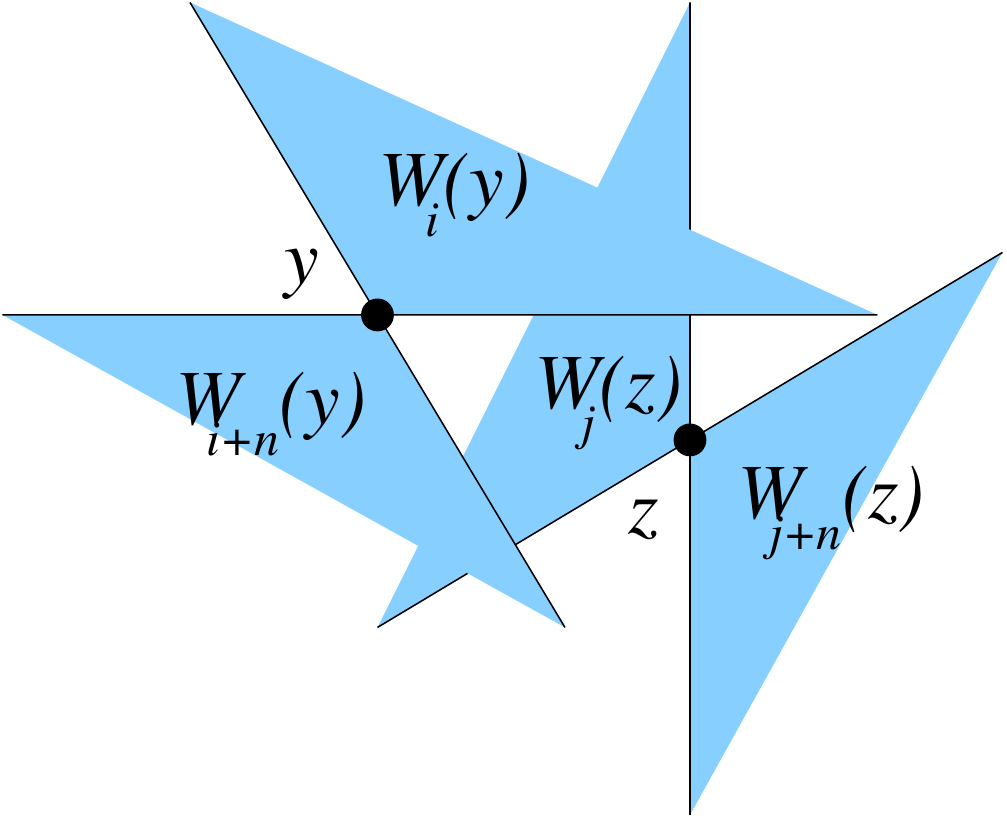}
\end{center}
\caption{\label{fig:pflemoneinter}Proof of Lemma~\ref{lem:oneinter}.}
\end{figure}

\begin{lemma}\label{lem:C_and_T}
If  $\bigc_i(r)$ intersects the interior of
$\bigt$, then it intersects the boundary of $\bigt$ at exactly two distinct lines.
We denote the lines by $\ell_{a_i}$ and $\ell_{b_i}$, so that $\ell_{a_{i}}$, $v_{i}$ and $\ell_{b_{i}}$ appear in counterclockwise order on the boundary of $\bigt$.
\end{lemma}
\begin{proof}

Take any point $x$ on $\bigc_i(r)\cap\bigt$. We have $v_i \in W_i(x)$ and $v_{i{+}n} \in W_{i{+}n}(x)$, and each of the common
supporting lines of those two wedges  properly intersects $\bigc_i$ only once. 
So each of the two wedges complementary to the union of $W_i(x)$ and $W_{i{+}n}(x)$ contains at least one
intersection of $\bigc_i$ with the boundary of $\bigt$. 
This implies that  $a_i \in [i{-}n,i{-}1]$ and $b_i\in [i,i{+}n{-}1]$.
Note that since  the property is valid for all $x$, intersections occur only on the lines $\ell_{a_{i}}$ and $\ell_{b_{i}}$.
\end{proof}

\noindent Lemma~\ref{lem:C_and_T} implies that every $\bigc_i(r)$ intersecting $\bigt$ is such that $i \in [a_i{+}1,b_i]$.
Let $\bigc'_i(r)$ be the portion of $\bigc_i(r)$ contained in $\bigt$:
$$
\bigc'_i(r) := \bigc_i(r)\cap\bigt .
$$

\begin{lemma}
(i) The curve $\bigc'_i(r)$ is connected. 
(ii) If $a_i\neq i{-}1$, then $\bigc'_j(r)$ is empty for $j \in [a_i{+}1, i{-}1]$. 
(iii) If $b_i \neq i$, then $\bigc'_j(r)$ is empty for $j \in [i{+}1,b_i]$.
\end{lemma}

\begin{proof}
Statement (i) follows directly from the fact that $\bigc_i(r)$
is an unbounded curve and intersects $\bigt$ at most twice
(Lemma~\ref{lem:C_and_T}). 
Statements (ii) and (iii) follow from Lemma~\ref{lem:C_and_T} and Lemma~\ref{lem:antipodal}.
\end{proof}

\begin{observation}
\label{obs:witnessvi}
If $\bigc'_i(r)$ is empty, then any $i$-wedge $W_i(x)$ for $x\in\bigt$
contains at least $r$ points of $S$. In particular, $W_i(v_i) \subseteq W_i(x)$ and $|W_i(v_i)\cap S|\geq r$.
\end{observation}

The combinatorial properties described in this section lay the foundations for the definition of a set of witness wedges in Section~\ref{sec:wedges}.

\section{Witness Wedges}
\label{sec:wedges}

We now describe a set of wedges, parameterized by a real number
$t\in[0,2n)$ with apex at point $x(t)$ and  
$type(t) = \lfloor t \rfloor$. We abbreviate $W_{type(t)}(x(t)) = W(t)$.
This set of wedges is such that any $i$-wedge containing at
least $3r{+}5$ points contains a witness wedge $W(t)$.
Thus it suffices to color only those witness wedges.

The wedge $W(t)$ will have its apex on $\bigc'_i(r)$ for $t\in
[i,i{+}1)$ if $\bigc'_i(r)$ is not empty.
More precisely, we let $\sigma_i(t)$, $t\in [i,i{+}1)$ be a parametrization of $\bigc'_i(r)$, 
where $\sigma_i(i) := \ell_{a_i}\cap\bigc'_i(r)$ and 
$\sigma_i(t) := \ell_{b_i}\cap\bigc'_i(r)$ for $t\in[i{+}0.9 ,i{+}1)$.
If $\bigc'_i(r)$ is empty, then we distinguish three cases (see Figure~\ref{fig:param}):
\begin{itemize}
\item [A.] If there is a $j$ such that $i\in [a_j{+}1,j{-}1]$ then 
  $\sigma_i(t) := \bigc'_j(r)\cap\ell_{a_j}$ for $t\in[i,i{+}1)$.
\item [B.] If there is a $j$ such that $i\in [j{+}1,b_j]$ then 
  $\sigma_i(t) := \bigc'_j(r)\cap\ell_{b_j}$ for $t\in[i,i{+}1)$.
\item [C.] Otherwise, $\sigma_i(t) := v_i$ for $t\in[i,i{+}1)$.
\end{itemize}
We define $x(t)$ as the concatenation of the functions $\sigma_i(t)$:
$$
x(t) := \sigma_{\lfloor{t}\rfloor}(t).
$$

\begin{figure}[htb]
\begin{center}
\includegraphics[scale=.7]{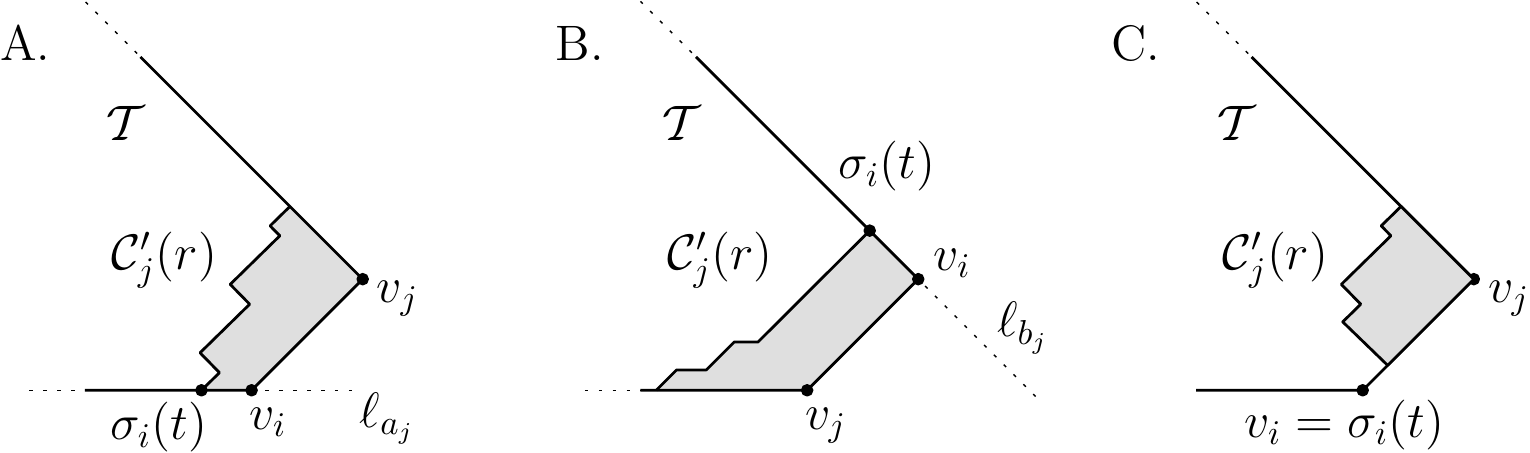}
\end{center}
\caption{\label{fig:param}Definition of $\sigma_{i}(t)$, when $\bigc'_i(r)$ is empty.}
\end{figure}

\begin{lemma}
\label{lem:witnesses}
For any wedge $W_i(y)$ that contains at least $3r{+}5$ points of $S$,
there is a value $t\in[i,i{+}1)$ such that $W(t)\subseteq W_i(y)$
and $W(t)$ contains at least $r$ points.
\end{lemma}
\begin{proof}
Since $W_i(y)$ contains at least $3r{+}5$ points, it must intersect $\bigt$. Thus it contains a wedge
$W_i(z)$ such that $z\in\bigt$. 

First suppose that $\bigc'_i(r)$ is not empty. Then $\bigc'_i (r)\cap W_i(z)\not= \emptyset$, otherwise
$W_i(z)$ cannot contain enough points. 
If $x(t)\in \bigc'_i (r)\cap W_i(z)$, then $W(t)$
is contained in $W_i(y)$ and contains at least $r$ points. 

Now suppose $\bigc'_i(r)$ is empty and refer to the three cases above. 
In case C, from Observation~\ref{obs:witnessvi}, $W_i(v_i) \subseteq W_i(z)$ and
$W_i(v_i)$ contains at least $r$ points. Since in that case $x(t) =
v_i$ for $t\in[i,i{+}1)$, any value of $t$ in $[i,i{+}1)$ will work. 
In case A, note that  wedges $W(i)$ and $W(j)$ have the same apex 
$\bigc'_j(r)\cap\ell_{a_j}$, $W(j)$ contains at most $r{+}1$ points,
$\ell_{a_j}$ has $2r{+}3$ points on its left, and $W(i)$ is in the
union of $W(j)$ and the halfplane left of $\ell_{a_j}$.
This implies that both $W(i)$ and the halfplane to the left of the oriented
line of direction $q_{i}q_{i{+}1}$ through its apex have at most $3r{+}4$
points. Thus $W_i(y)$ has its apex to the right
of that line, which implies $W(i)\subseteq W_i(y)$. 
Because $\bigc'_i(r)$ is empty, $|W(i)\cap S|\geq r$.
Case B is identical.
\end{proof}

It is natural to view the range $[0,2n)$ as a counterclockwise
parametrization of the points on a unit circle. Thus in what follows,
the real parameter $t$ will be viewed modulo $2n$, and an interval
$[t,t']$ is the set of points on the circle on a counterclockwise walk
from $t$ to $t'$. 

\section{Reduction to Intervals}
\label{sec:inter}

Our goal is to color the points of $S$ with $k$ colors such that any witness wedge $W(t)$ contains at least one point of each color.
For each point $p$ in $S$, we consider the set $I(p)$ of witness wedges containing $p$:
$$
I(p) := \{ t\in [0,2n) : p\in W (t) \} .
$$

\begin{lemma}
\label{lem:oneinterval}
For any point $p\in S$, if $p\in W(t) \cap W(t')$, where $t$ appears
before $t'$ (that is,  $t'\notin [\lfloor t \rfloor,t]$ and 
$type(t')\in [type(t),type(t){+}n{-}1]$), then $p\in W(t'')$ for all $t''\in[t,t']$.
\end{lemma}
\begin{proof}
There are several cases to consider. If $type(t)=type(t')$, then either
$W(t)=W(t')$, or $x(t'')$ lies on $\bigc'_i(r)$ between $x(t)$
and $x(t')$. Since $\bigc'_i(r)$ is  monotone in all directions between
$q_iq_{i{-}1}$ and $q_iq_{i{+}1}$, the wedge $W(t'')$ contains the
intersection of $W(t)$ and $W(t')$.

In the second case, $type(t') \in [type(t){+}1,type(t){+}n{-}1]$. 
Then by Lemma~\ref{lem:antipodal}, $p\notin \bigt$. Also, because
$x(t)$ and $x(t')$ are in $\bigt$ and by the same lemma, we know
$x(t)\notin W(t')$ and $x(t')\notin W(t)$. This implies that the
counterclockwise bounding ray $u$ of $W(t)$ intersects the clockwise
bounding ray $u'$ of $W(t')$ and $p$ is in the closed wedge $V$ right
of $u$ and left of $u'$ (see Figure~\ref{fig:lemma8}). Then for any point $q$ in the closed wedge
$V'$ opposite to $V$, the wedges $W_j(q)\supset V$ for
$j\in[type(t),type(t')]$.  

\begin{figure}[htb]
\begin{center}
\includegraphics{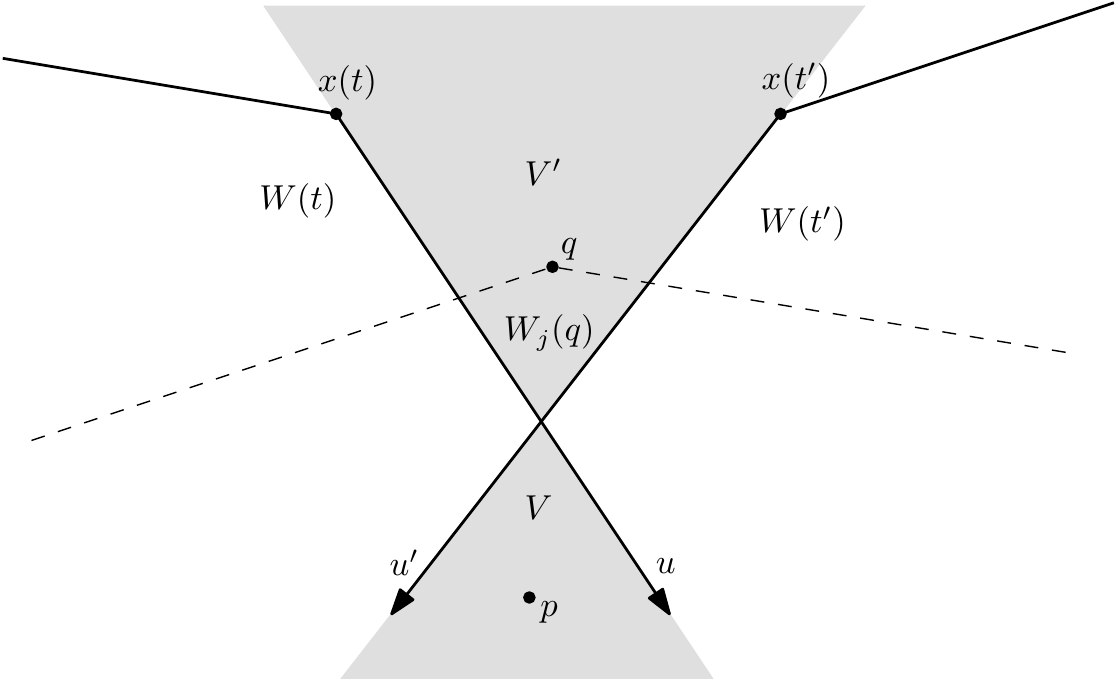}
\end{center}
\caption{\label{fig:lemma8}Second case of Lemma~\ref{lem:oneinterval}.}
\end{figure}

In the remainder of this proof, we will show that $p\in W(t'')$ for
all $t''\in [t,t']$ where $t''$ is an integer or an integer plus
$0.9$. Then the lemma follows by applying the first case for every
other value of $t''$.
In fact, it will suffice to show that $p\in W(\lfloor t \rfloor {+}0.9)$
and $p\in W(\lfloor t \rfloor {+}1)$ (or symmetrically that 
$p\in W(\lceil t' \rceil {-}1.1)$ and $p\in W(\lceil t' \rceil {-}1)$) and
apply the lemma again on the subrange $[t,t'']$ (or $[t'',t']$).

If $x(t)$ lies on $\bigc'_{type(t)}(r)$, then so does 
$x(\lfloor t\rfloor{+}0.9) = x(\lfloor t\rfloor{+}1)$. 
So $W(\lfloor t\rfloor{+}0.9)$ contains at most $r{+}1$ points. This
implies, by the same argument as above, that the ray $u$ intersects
the clockwise bounding ray of $W(\lfloor t\rfloor{+}0.9)$ and that ray
$u'$ intersects the counterclockwise bounding ray of 
$W(\lfloor t\rfloor{+}0.9)$. Therefore, 
$x(\lfloor t\rfloor{+}0.9)\in V'$ 
(and so $x(\lfloor t\rfloor{+}1)\in V'$), which implies 
$p\in W(\lfloor t\rfloor{+}0.9)$ and
$p\in W(\lfloor t\rfloor{+}1)$. 
The case where $x(t')$ lies on $\bigc'_{type(t')}(r)$ is covered symmetrically.

If $\bigc'_{type(t)}(r)$ is empty then $W(\lfloor t \rfloor {+}0.9)=W(t)$. 
Furthermore, if $x(t)$ is defined according to case A,
$x(\lfloor t \rfloor {+}1)=x(t)\in V'$ thus 
$p\in W(\lfloor t \rfloor {+}1)$.
If $x(t)$ is defined according to case B, then either 
$x(\lfloor t \rfloor {+}1)=x(t)$ and we are done, or $type(t)=b_j$ which
we treat below. 

Finally, we are left with the case where $x(t)$ is defined according
to case C or it is defined according to case B and $type(t)=b_j$. By
symmetry, we also assume that $x(t')$ is either defined according to
case C or it is defined according to case A and $type(t')=a_j{+}1$.
Note that in this case, the entire portion of the boundary of $\bigt$
between $x(t)$ and $x(t')$ is inside $V'$. This implies again that 
$p\in W(\lfloor t \rfloor {+}1)$.\end{proof}

As a consequence, a point defines either an interval, or a pair of intervals, the corresponding wedges of which are of two types $i$ and $i{+}n$.

\begin{corollary}
\label{cor:interval}
$I(p)$ is either an interval, or a pair of intervals $I_1(p) ,I_2(p)$, such that $type (t)=i$ for $t\in I_1(p)$ and $type(t)=i{+}n$ for $t\in I_2(p)$, where $i$ is such that $\bigw_i^{<r}$ and $\bigw_{i{+}n}^{<r}$ intersect in $\bigt$.
\end{corollary}
\begin{proof}
From Lemma~\ref{lem:oneinterval}, $I(p)$ cannot consist of more than two intervals, since otherwise we can find two points $t$ and $t'$ satisfying the conditions of Lemma~\ref{lem:oneinterval} in two distinct intervals.

Now first suppose  that no pair $\{\bigw_i^{<r},\bigw_{i{+}n}^{<r}\}$ intersect in $\bigt$. Then again the statement is a direct consequence of Lemma~\ref{lem:oneinterval}. Otherwise suppose that $p\in \bigw_i^{<r}\cap \bigw_{i{+}n}^{<r}$. Then we must show that $I_{1}(p)\subset [i,i{+}1)$ and $I_{2}(p)\subset [i{+}n,i{+}n{+}1)$. For contradiction, let $i{+}1$ be contained strictly in the interior of $I_{1}(p)$. Then there are again two points $t\in I_{1}(p)\cap [i{+}1,i{+}2)$ and $t'\in I_{2}(p)\cap [i{+}n,i{+}n{+}1)$ satisfying the conditions of Lemma~\ref{lem:oneinterval}, a contradiction.
\end{proof}

\section{Coloring}
\label{sec:coloring}

We give an algorithm for coloring the points with $k$ colors so that all wedges $\{ W(t) : t\in [0,2n)\}$ contain all $k$ colors. In the following, we say that a point $p\in S$ {\em covers} a point $t\in [0,2n)$ whenever $t\in I (p)$. We proceed by iteratively removing a covering of $[0,2n)$, that is, a subset of $S$, the elements of which collectively cover the circle $[0,2n)$. We use a greedy algorithm to select such a subset; we iteratively expand the cover for $\left[0,t\right)$, by selecting a new point that covers the largest interval starting from $t$. Every point in a cover is assigned the same color. By repeating this $k$ times, we ensure that all $k$ colors are represented in each of the wedges $W(t)$, and thus by Lemma~\ref{lem:witnesses}, in all wedges containing at least $3r{+}5$ points. The key property of the algorithm is that it only requires $r=O(k)$.

A formal description of the algorithm follows. We suppose, without loss of generality, that only the pair $\{\bigw_{0}^{<r}, \bigw_{n}^{<r}\}$ may intersect in $\bigt$.\medskip

\noindent \begin{minipage}{\textwidth}
\noindent {\bf Coloring Algorithm}\\

\noindent for $i\gets 1$ to $k$ do:
\begin{enumerate}
  \item $x\gets 0$, $S'\gets \emptyset$
  \item while $\bigcup_{p\in S'} I(p)\not= [0,2n)$ do:
  \begin{enumerate}
     \item find $p\in S 
     $ such that $y(p) := \max_{t\in [0,2n)} \{ t{-}x : [x,t]\subseteq I(p) \}$ is maximized
     \item $S'\gets S' \cup \{ p \}$
     \item $x\gets x {+} y(p)$
  \end{enumerate}
  \item assign color $i$ to all points in $S'$
  \item $S\gets S \setminus S'$
\end{enumerate}
~
\end{minipage}

When every set $I(p)$ is a simple interval, this algorithm greedily colors circular arcs. The following lemma states that in that case, no point on a circle is covered more than a constant number of times per iteration (see Figure~\ref{fig:arcs}).

\begin{figure}[htb]
\begin{center}
\includegraphics[scale=.4, angle=-90]{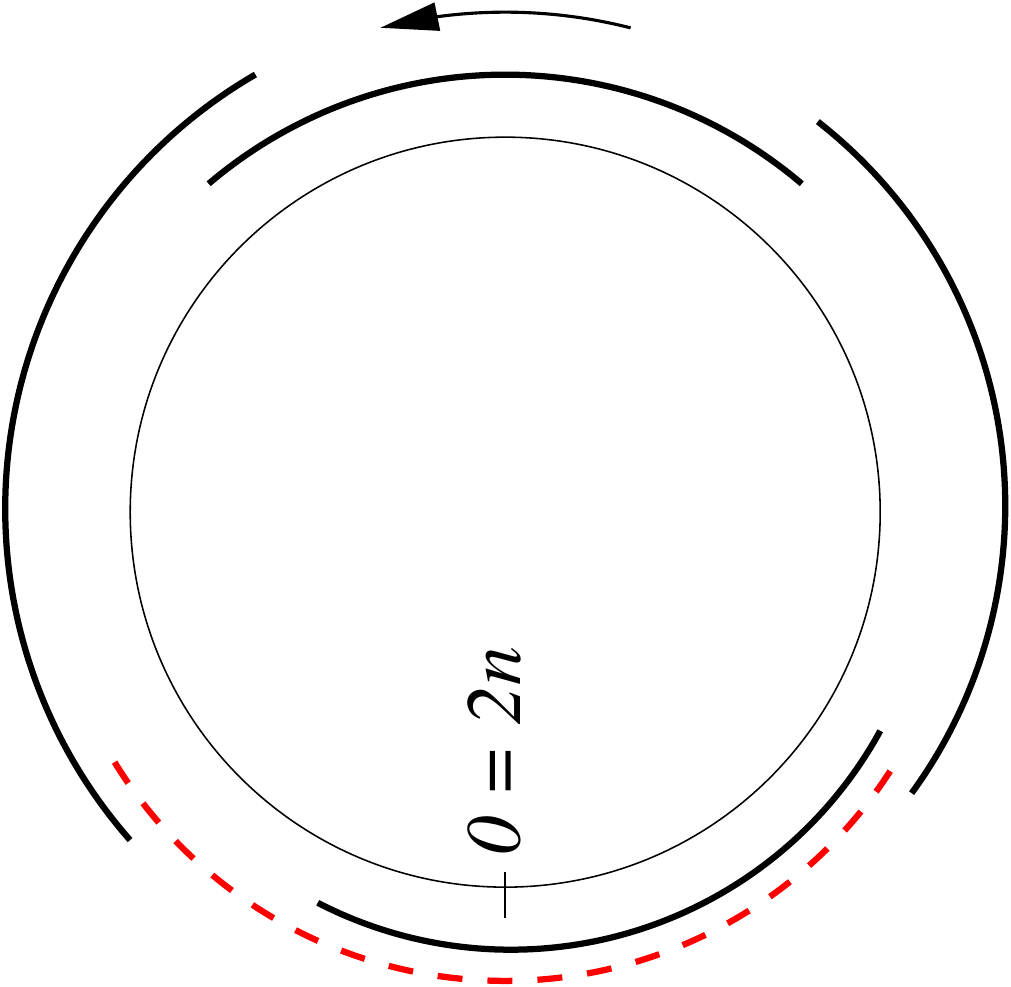}
\end{center}
\caption{\label{fig:arcs}Covering the circle $[0,2n)$ by circular arcs.}
\end{figure}

\begin{lemma}
\label{lem:coloring}
Suppose that no pair $\{\bigw_i^{<r}, \bigw_{i{+}n}^{<r}\}$ intersect in $\bigt$, and that there are enough points to perform $j$ iterations of the coloring algorithm. Let $V$ be the set of points colored by the algorithm after the iteration $j$. Then every point of $[0, 2n)$ is covered at most $3j$ times by points of $V$.
\end{lemma}
\begin{proof}
It is sufficient to prove that no point of $[0,2n)$ is covered more than three times by points of $S'$.
Note that if no pair of curves intersect, then from Corollary~\ref{cor:interval}, every set $I(p)$ is an interval. Hence $S'$ is a greedy covering of the circle by intervals (i.e., circular arcs).

Let $I(p)$ be the last interval chosen by the algorithm, and consider $S'' := S'{-}\{p\}$. Suppose that a point $t$ is covered by more than two points of $S''$. Let $a$ and $b$ be the first and the last points chosen, respectively, that cover $t$. The remaining intervals that cover $t$ either extend further than $I(b)$ and should have been chosen instead of $I(b)$, or do not extend further than $I(b)$, in which case $I(b)$ should have been chosen instead. In both cases, we have a contradiction. Hence the points of $S''$ do not cover any point of $[0,2n)$ more than twice. The last interval $I(p)$ can cover some points of the circle a third time. Therefore, every point of $[0,2n)$ is covered at most three times by points of $S'$.
\end{proof}

In the general case, a point $p$ might correspond to two intervals on opposite regions of the circle $[0,2n)$. We show that the following similar property holds.
\begin{lemma}
\label{lem:gen-coloring}
Suppose that there are enough points to perform $j$ iterations of the above algorithm, and let $V$ be the set of points colored by the algorithm after the iteration $j$. Then every point of $[0, 2n)$ is covered at most $6j$ times by points of $V$.
\end{lemma}
\begin{proof}
We consider that $\bigw_0^{<r}$ and $\bigw_n^{<r}$ intersect in $\bigt$. Otherwise, the statement is implied by Lemma~\ref{lem:coloring}. By Lemma~\ref{lem:oneinter}, only one such pair can intersect. Without loss of generality, we also assume that $\bigc'_0$ and $\bigc'_n$ are both orthogonal staircases going from top left to bottom right. This setting can always be enforced by symmetry and affine transformation of the points. We assume that $\bigc'_0$ is, at some point, above $\bigc'_n$, which might cause a point between the two curves to generate one interval on each (see Figure~\ref{fig:twointervals}).

We will prove our statement by induction on the number of iterations. Let us show that after the iteration $(j{+}1)$, no point of $[0,2n)$ is covered more than $6(j{+}1)$ times. The induction hypothesis is that this is true for the iterations 0 to $j$, where iteration $0$ corresponds to the initial situation. The base case $j=0$ is trivial.

Consider a point $t\in [n, n{+}1)$, and the corresponding point $x=x(t)$ on  $\bigc'_n$. Suppose that this point was covered $\tau$ times in the previous iterations (thus by points of colors 1 to $j$). By the induction hypothesis, $\tau \leq 6j$.
We consider the set of points $S'$ selected by the algorithm at the iteration $(j{+}1)$. The sets $I(p)$ start by covering the wedges of type 0, corresponding to points on $\bigc'_0$. Let $p$ be the first point of $S'$, in order of selection, that also covers $t$. By Corollary~\ref{cor:interval}, $p$ only covers two types of wedges, 0 and $n$. Let $q$ be the horizontal projection of $p$ on $\bigc'_0$.

Let $p'$ be the next point selected by the greedy algorithm. If it covers the point $1$, then from Corollary~\ref{cor:interval}, it cannot cover any point on $\bigc'_n$. Otherwise, since the algorithm is greedy, the point is associated with an interval that intersects $I(p)$ and that has the farthest right endpoint. Geometrically, $p'$ is the lowest point to the left of the vertical line $\ell$ through $q$. Let $z$ be the projection of $x$ onto $\ell$.
   
Two cases can occur. First, if $p'$ is below $x$, then $t$ is covered at most once, by $p$.
On the other hand, if $p'$ is above $x=x(t)$, then $p'$ covers $t$. 

By the  induction hypothesis,  $W_0 (q)$ contains at most $6j$ colored points. By Observation~\ref{obs:rorrp1} and since $q\in\bigc'_0$,  $W_0 (q)$ contains at least $r$ points. Hence $W_0 (q)$ contains at least $r{-}6j$ uncolored points (including $p$ and $p'$). Also, since the algorithm is greedy, $W_0 (p)$ and $W_0(z)$ do not contain  uncolored points, otherwise they would have been selected by the algorithm. Hence the orthogonal rectangle $R$ with opposite vertices $p$ and $z$  contains at least $r{-}6j$ uncolored points (see Figure~\ref{fig:proof}).

\begin{figure}[htb]
\begin{center}
\subfigure[\label{fig:twointervals}The point $p$ is associated with two intervals $I_1(p)\subseteq [0,1)$, and $I_2(p)\subseteq [n,n{+}1)$.]{\includegraphics[angle=-90, scale=.5]{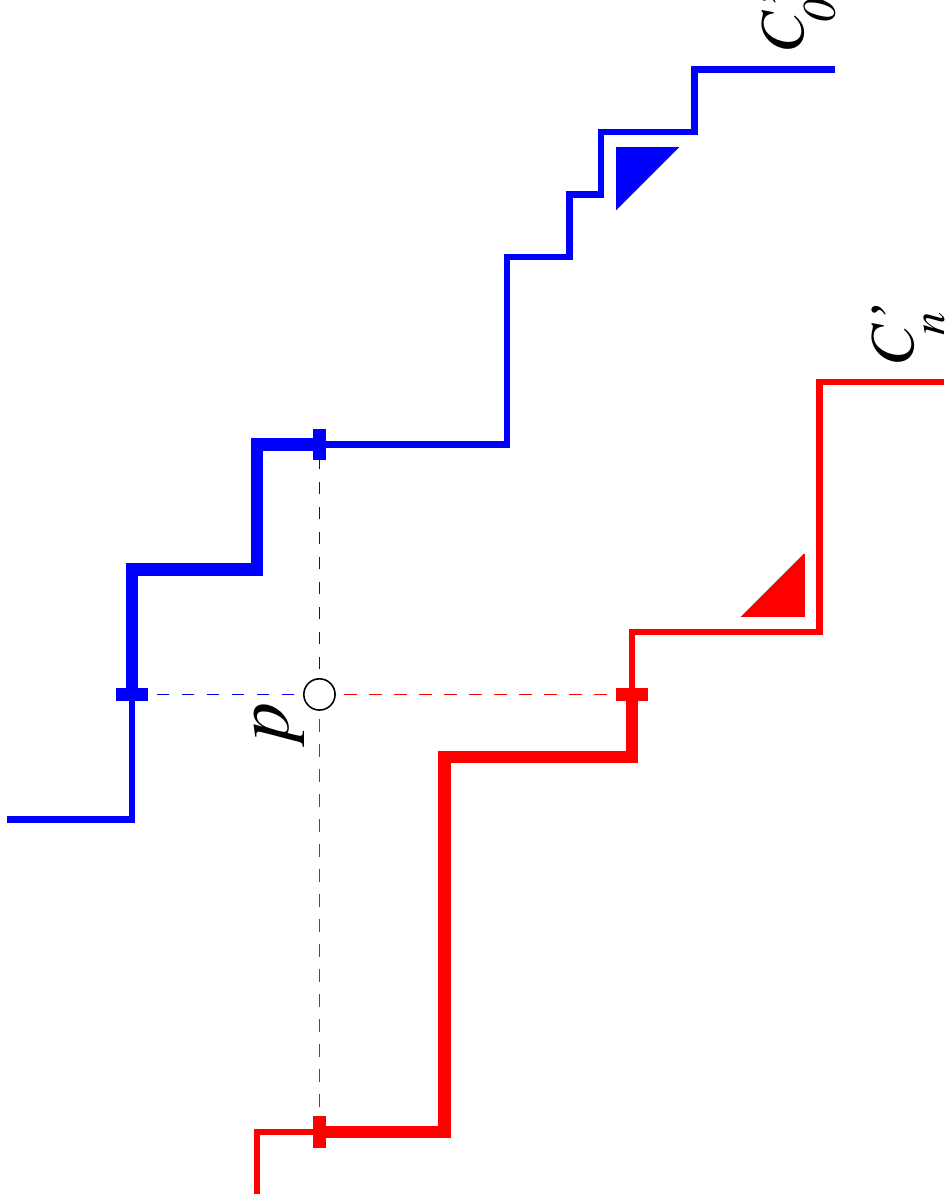}}
\hspace{1.5cm}
\subfigure[\label{fig:proof}Definition of the region $R$.]{\includegraphics[angle=-90, scale=.5]{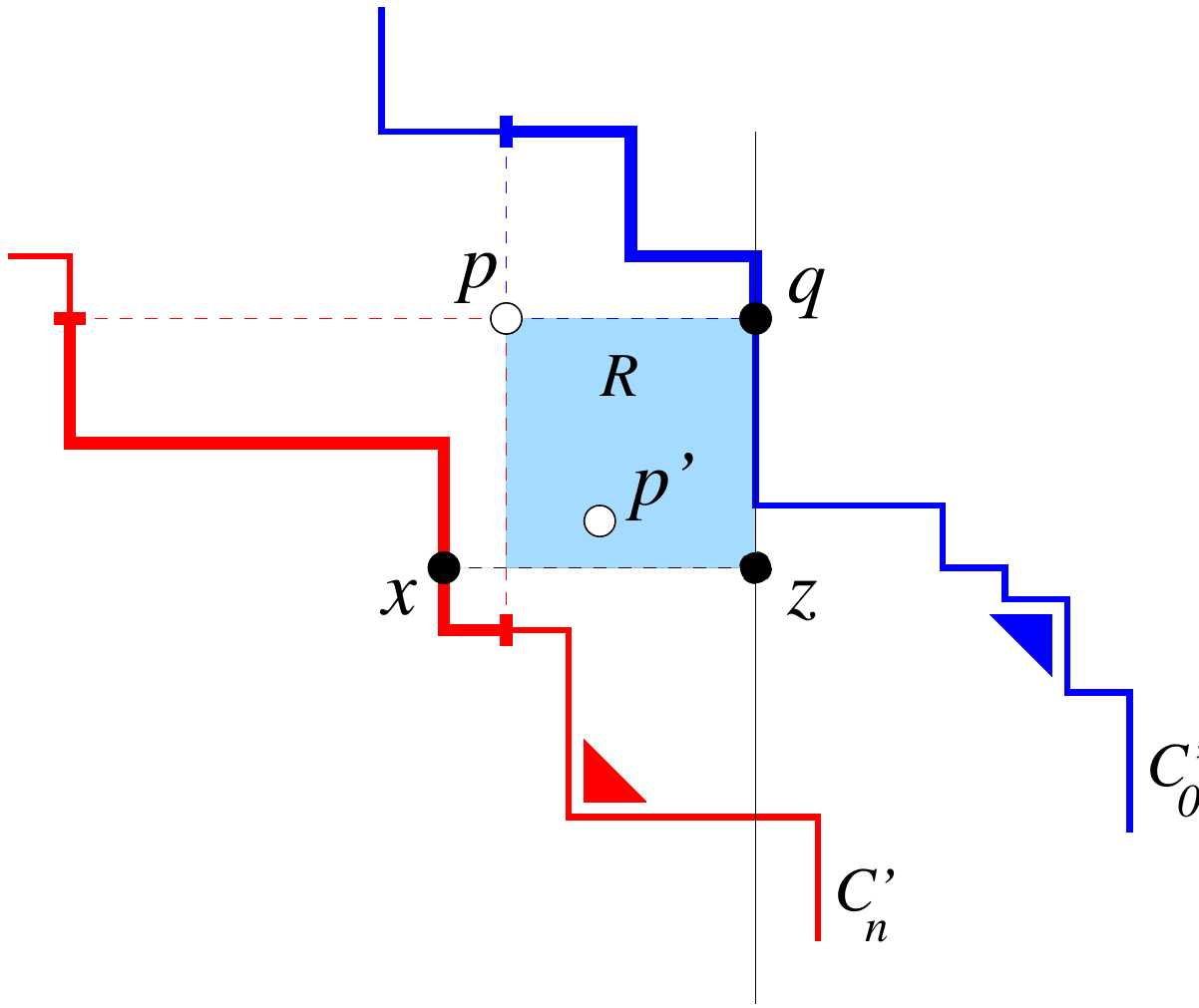}}
\end{center}
\caption{Illustration of the proof of Lemma~\ref{lem:gen-coloring}.}
\end{figure}

Since $t$ is covered $\tau$ times and $x=x(t)\in\bigc'_n$,  $W_n (x)$ can contain at most $r{+}1{-}\tau$ uncolored points. $R$ is included in $W_n(x)$, thus from the previous observation on $R$, there are at most $(r {+} 1 {-} \tau) {-} (r {-} 6j) = 6j {+} 1 {-} \tau$ uncolored points that are both to the right of $\ell$ and above $x$. These, together with $p$ and $p'$, are the only points that may cover $t$ after we have covered the interval $[0,1)$. Hence after we have covered the interval $[0,1)$, the points in $[n,n{+}1)$ cannot be covered more than $\tau {+} (6j {+} 1 {-} \tau) {+} 2 = 6j {+} 3$ times. On the other hand, the points in $[0,1)$ cannot be covered more than $6j{+}3$ times.

A similar reasoning holds when the algorithm starts to cover points in the interval $[n,n{+}1)$. We have to replace $6j$ by $6j{+}3$, since the points on both sides can already be covered $6j{+}3$ times. Thus after the  iteration $(j{{+}}1)$, no point is covered more than $6j{+}3{+}3 = 6(j{{+}}1)$ times, which concludes the proof.
\end{proof}

\begin{corollary}
For $r\geq 6k$, the coloring algorithm finds a $k$-coloring of the points in $S$ such that all wedges $W(t)$ for $t\in [0,2n)$ contain all $k$ colors.
\end{corollary}

Note that with this choice of $r$, by the well-known center point theorem, $\bigt$  is never empty. By Lemma~\ref{lem:witnesses}, this concludes the proof of Theorem~\ref{thm:rephrased}, with $\alpha_Q k \geq c_Q \times (3 \times 6k {+} 5)$, hence for any $\alpha_Q \geq 23 c_Q$. By duality, as the polygon $Q$ is symmetric, this implies the following.

\begin{theorem}
\label{thm:main}
Given a centrally symmetric convex polygon $Q$, there exists a constant $\alpha_Q$ such that every $\alpha_Q k$-fold covering of the plane by translates of $Q$ can be decomposed into $k$ coverings.
\end{theorem}

\bibliography{coverings}

\begin{thebibliography}{10}

\bibitem{faulttolerant}
M.~Ali Abam, M.~de~Berg, and S.~Poon.
\newblock Fault-tolerant conflict-free coloring.
\newblock In {\em Proceedings of the Canadian Conference on Computational
  Geometry ({CCCG'08})}, 2008.

\bibitem{jithamilton}
M.~Abellanas, P.~Bose, J.~Garcia, F.~Hurtado, C.~Nicolas, and P.~Ramos.
\newblock On properties of higher order delaunay graphs with applications.
\newblock In {\em Proceedings of the European Workshop on Computational
  Geometry ({EWCG'05})}, 2005.

\bibitem{ACCLS08}
G.~Aloupis, J.~Cardinal, S.~Collette, S.~Langerman, and S.~Smorodinsky.
\newblock Coloring geometric range spaces.
\newblock In {\em Proceedings of the 8th Latin American Theoretical Informatics
  ({LATIN}'08)}, Lecture Notes in Computer Science. Springer-Verlag, 2008.

\bibitem{RPDG}
P.~Brass, W.~O.~J. Moser, and J.~Pach.
\newblock {\em Research Problems in Discrete Geometry}.
\newblock Springer-Verlag, 2005.

\bibitem{othersensors}
A.~Buchsbaum, A.~Efrat, S.~Jain, S.~Venkatasubramanian, and K.~Yi.
\newblock Restricted strip covering and the sensor cover problem.
\newblock In {\em Proceedings of the ACM-SIAM Symposium on Discrete Algorithms
  (SODA'07)}, 2007.

\bibitem{settheo}
M.~Elekes, T.~Matrai, and L.~Soukup.
\newblock On splitting infinite-fold covers.
\newblock Unpublished manuscript, 2007.

\bibitem{shakharcf}
G.~Even, Z.~Lotker, D.~Ron, and S.~Smorodinsky.
\newblock Conflict-free colorings of simple geometric regions with applications
  to frequency assignment in cellular networks.
\newblock {\em {SIAM} Journal on Computing}, 33(1):94--136, 2004.

\bibitem{MP86}
P.~Mani and J.~Pach.
\newblock Decomposition problems for multiple coverings with unit balls.
\newblock Unpublished manuscript, 1986.

\bibitem{Pa80}
J.~Pach.
\newblock Decomposition of multiple packing and covering.
\newblock In {\em 2. Kolloq. {\"u}ber Diskrete Geom.}, pages 169--178. Inst.
  Math. Univ. Salzburg, 1980.

\bibitem{Pach86}
J.~Pach.
\newblock Covering the plane with convex polygons.
\newblock {\em Discrete {\&} Computational Geometry}, 1:73--81, 1986.

\bibitem{PTT07}
J.~Pach, G.~Tardos, and G.~T\'{o}th.
\newblock Indecomposable coverings.
\newblock In {\em The China--Japan Joint Conference on Discrete Geometry,
  Combinatorics and Graph Theory ({CJCDGCGT}'05)}, Lecture Notes in Computer
  Science, pages 135--148, 2007.

\bibitem{PT07}
J.~Pach and G.~T\'{o}th.
\newblock Decomposition of multiple coverings into many parts.
\newblock In {\em Proceedings of the 23rd {ACM} Symposium on Computational
  Geometry ({SoCG}'07)}, pages 133--137, 2007.

\bibitem{TT07}
G.~Tardos and G.~T\'oth.
\newblock Multiple coverings of the plane with triangles.
\newblock {\em Discrete {\&} Computational Geometry}, 38(2):443--450, 2007.

\end{thebibliography}
\bibliographystyle{plain}

\end{document}